\documentclass[journal]{IEEEtran}
\usepackage[noadjust]{cite}
\usepackage{enumitem}
\usepackage{amsmath,amssymb,amsthm}
\newtheorem{thm}{Theorem}[section]
\newtheorem{lemma}[thm]{Lemma}
\newtheorem{prop}[thm]{Proposition}

\newtheorem{cor}[thm]{Corollary}

\newtheorem{exampl}[thm]{Example}

\newtheorem{remark}[thm]{Remark}
\DeclareMathOperator{\ind}{ind}

\DeclareMathOperator{\ord}{ord}
\begin{document}
\title{On the Existence and Nonexistence of Splitter Sets}
\author{
	\IEEEauthorblockN{Zhiyu Yuan\IEEEauthorrefmark{1}, Rongquan Feng\IEEEauthorrefmark{2} and Gennian Ge\IEEEauthorrefmark{3}}

	\IEEEauthorblockA{\IEEEauthorrefmark{1}School of Mathematical Sciences, Peking University, Beijing 100871, China. Email: yzhiyu\_pku@pku.edu.cn}

	\IEEEauthorblockA{\IEEEauthorrefmark{2}School of Mathematical Sciences, Peking University, Beijing 100871, China. Email: fengrq@math.pku.edu.cn}

	\IEEEauthorblockA{\IEEEauthorrefmark{3}School of Mathematical Sciences, Capital Normal University, Beijing 100048, China. Email: gnge@zju.edu.cn}

}
\maketitle{}
\begin{abstract}
	In this paper, the existence of perfect and quasi-perfect splitter sets in finite abelian groups is studied, motivated by their application in coding theory for flash memory storage. For perfect splitter sets we view them as splittings of $\mathbb{Z}_n$, and using cyclotomic polynomials we derive a general condition for the existence of such splittings under certain circumstances. We further establish a relation between $B[-k, k](q)$ and $B[-(k-1), k+1](q)$ splitter sets, and give a necessary and sufficient condition for the existence of perfect $B[-1, 5](q)$ splitter sets. Finally, two nonexistence results for quasi-perfect splitter sets are presented.
\end{abstract}
\section{Introduction}
Recent years have witnessed the development of flash memory, which highlights the importance of studying codes for flash memory storage. For this application, a new model of limited magnitude error-correcting was proposed (\cite{cassuto2010codes},\cite{klove2011systematic}). In this model, every symbol is a number in $\mathbb{Z}_q=\{0,1,\ldots,q-1\}$, the set of integers modulo $q$, where $q$ is a positive integer, and errors are considered to be bounded changes on certain coordinates. To be specific, in this model, messages are encoded into words in $\mathbb{Z}_q^n$, and during transmission of a codeword $\mathbf{c}=(c_1,\ldots,c_n)$, $c_i$ can by changed to $c_i+\lambda$, where $\lambda\in[-k_1,k_2]:=\{-k_1,-k_1+1,\ldots,k_2\}$ is an integer for some nonnegative integer $k_1$ and positive integer $k_2$. In this paper, we consider only \emph{single} limited magnitude errors, i.e., such change only happens on at most one coordinate of a codeword.

We define the fundamental error ball $\mathcal{E}(k_1;k_2)$ as
\[\mathcal{E}(k_1;k_2)=\bigcup_{i=1}^n (\{0\}^{i-1}\times [-k_1,k_2]\times \{0\}^{n-i}),\]
that is, for every element of $\mathcal{E}(k_1;k_2)$, all its coordinates are zero, except for possibly one nonzero coordinate that is in $[-k_1,k_2]$. Then our goal is to construct a code $C\subset\mathbb{Z}_q^n$ such that for all $\mathbf{x}\in C$, $\mathbf{x}+\mathcal{E}(k_1;k_2)$ are disjoint. It is equivalent to saying that for all $\mathbf{e}\in \mathcal{E}(k_1;k_2)$, $\mathbf{e}+C$ are disjoint. We call such codes \emph{single limited magnitude error-correcting codes}.

For an element $s$ in an abelian group $(G,+)$ there is a homomorphism $\phi_s\colon \mathbb{Z}\to G$ such that $\phi_s(1)=s$. By abuse of notations, we will write $\lambda s$ for $\phi_s(\lambda)$ in this paper.

We assume that the code we are considering, $C$, is a subgroup of $(\mathbb{Z}_q^n,+)$. Then we have the canonical map
\[\mathbb{Z}_q^n \to \mathbb{Z}_q^n/C=:G,\]
where maps $\mathbf{x}$ to $\mathbf{x}+C$.
Let $s_i$ be the image of $\mathbf{e}_i=(\underbrace{0,\ldots,0}_{i-1},1,0,\ldots,0)$, then $C$ is a limited magnitude error-correcting code if and only if $\lambda s_i$ are all different, where $\lambda\in[-k_1,k_2]^*=[-k_1,k_2]\setminus\{0\}$, $i\in [1,n]$.

From now on we suppose that $G=\mathbb{Z}_N$ is a cyclic group. If $\lambda s_i, \lambda\in[-k_1,k_2]^*$ are all different, we say that $B=\{s_i:i=1,\ldots,n\}$ is a \emph{$B[-k_1,k_2](N)$ splitter set}.
It is easy to see that we can construct single limited error correcting codes of block length $n$ through splitter sets of size $n$.
Obviously, $|B|\le\frac{N-1}{k_1+k_2}$. If $|B|$ obtains its maximum $\left\lfloor\frac{N-1}{k_1+k_2}\right\rfloor$, we call it \emph{perfect} if $(k_1+k_2)\mid (N-1)$ or \emph{quasi-perfect} if $(k_1+k_2)\nmid (N-1)$. We call a $B[-k_1,k_2](N)$ splitter set \emph{nonsingular} if $\gcd(N,k_1!k_2!)=1$, or \emph{singular} otherwise.

Perfect splitter sets are relevant to a classical topic in algebra and combinatorics called \emph{group splitting}, see \cite{hickerson1983splittings, stein1967factoring,stein1972symmetric,
	szabo2009factoring}. Let $M$ be a finite set of nonzero integers and $(G,+)$ a finite abelian group.
We say that $G$ has a splitting $G\setminus\{0\}=MS$ if there exists a subset $S\subset G$ such that every nonzero element of $G$ can be uniquely represented as $ms, m\in M, s\in S$. In this case, we say $M$ is the \emph{multiplier set} of the splitting and $S$ its \emph{splitting set}. Clearly, perfect $B[-k_1,k_2](N)$ splitter sets are exactly splitting sets of splittings of $\mathbb{Z}_N$ with multiplier set $[-k_1,k_2]^*$. It is known that to study nonsingular perfect splitter sets $B[-k_1,k_2](N)$, we only need to consider the case when $N$ is a prime $q$ (\cite{ye2020some}).

Since the concept of splitter sets was defined in \cite{klove2011some}, many researchers have worked on them, for example, \cite{klove2012codes}, \cite{yuan2020existence} and \cite{zhang2018nonexistence}. Most recent results on the necessary and sufficient conditions for the existence of nonsingular perfect splitter sets are concluded in \cite{ye2020some}. These preceding results are mainly obtained by analyzing the group $\mathbb{Z}_q^\times$ and the structure of the splitter set $B$ case by case.

In this paper, we use cyclotomic polynomials to derive a necessary and sufficient condition for the existence of group splitting in certain cases, which generalizes a theorem of \cite{zhang2018nonexistence}. Then we show that this theorem covers known results and can produce some new existence theorems. We also show a relation between perfect $B[-k,k](q)$ and $B[-(k-1),k+1](q)$ splitter sets, and guided by a similar idea, we get the necessary and sufficient condition for the existence of perfect $B[-1,5](q)$ splitter sets. At last, we provide two nonexistence results on quasi-perfect splitter sets.
\section{Preliminary}
In this section, we give some terminology and lemmas used throughout this paper.
For any integer $a$ and any cyclic group $\mathbb{Z}_n$, we will always identify $a$ with its image $\phi_1(a)$ in $\mathbb{Z}_n$.
\begin{lemma}
	\label{lem-in}
	For every subgroup $H$ of a cyclic group $\mathbb{Z}_n$, $a\in H$ if and only if $[\mathbb{Z}_n:H]\mid a$.
\end{lemma}
\begin{proof}
	Let $i=[\mathbb{Z}_n:H]$, then it is clear that $i\mid n$, hence
	$$
		H':=\{0,i,2i,\ldots,n-i\}
	$$
	is a subgroup of $\mathbb{Z}_n$ of order $n \over i$ and of index $i$.
	Since $\mathbb{Z}_n$ is cyclic, for every divisor $d$ of $n$ there is a unique subgroup of $\mathbb{Z}_n$ of index $d$, so $H=H'$, and hence $a\in H=H'$ if and only if $i\mid a$.
\end{proof}
For a prime $p$, we fix one of its primitive roots $g$, and denote by $\ind_g(n)$ the minimal nonnegative integer $e$ such that $g^e\equiv n\pmod p$. The multiplicative order of $n\in\mathbb{Z}_p^\times$ is denoted by $\ord_p(n)$, and we have $|\langle n\rangle|=\ord_p(n)$, where $\langle n\rangle$ is the multiplicative subgroup of $\mathbb{Z}_p^\times$ generated by $n$. Let $v_p(n)$ be the $p$-adic valuation of $n$, i.e., $v_p(n)=\max\{e : p^e\mid n\}$.

\begin{lemma}
	Let $p$ be an odd prime, $e\mid (p-1)$, then $a$ is an $e$-th power residue modulo $p$ if and only if
	$e\mid [\mathbb{Z}_p^\times:\langle a\rangle]$.
\end{lemma}
\begin{proof}
	We know that $\mathbb{Z}_p^\times$ is a cyclic group of order $p-1$. Fix a primitive root $g$ modulo $p$, then $$\mathbb{Z}_p^\times=\{g^i:i\in[0,p-2]\}.$$ By the lemma above, $a=g^{\ind_g(a)}\in\langle a\rangle$ implies $[\mathbb{Z}_p^\times:\langle a\rangle]\mid\ind_g(a)$. If $e\mid [\mathbb{Z}_p^\times:\langle a\rangle]$, then $e\mid \ind_g(a)$ and
	$g^{\ind_g(a)/e}$ is an $e$-th root of $a$ modulo $p$.
	On the other hand, suppose that $a\equiv b^e\pmod{p}$. By the lemma above, $g^{[\mathbb{Z}_p^\times:\langle a\rangle]}\in\langle a\rangle$, hence there exists $i\in\mathbb{Z}$ such that $g^{[\mathbb{Z}_p^\times:\langle a\rangle]}=a^i$, and
	$[\mathbb{Z}_p^\times:\langle a\rangle]\equiv \ind_g(a)i\equiv
		\ind_g(b)ei\pmod{(p-1)}$. Since $e\mid(p-1)$, $e\mid [\mathbb{Z}_p^\times:\langle a\rangle]$.
\end{proof}

For a finite abelian group $(G,+)$ and its subsets $A_1,\ldots,A_n$, we define
$$
	A_1+A_2+\cdots+A_n=\{a_1+a_2+\cdots+a_n:a_i\in A_i\}.
$$
If sums $a_1+\cdots+a_n, a_1\in A_1,\ldots,a_n\in A_n$ are distinct and $G=A_1+\cdots+A_n$, we say $G=A_1+\cdots+A_n$ is a \emph{factorization} of $G$, and every $A_i$ is called its \emph{direct factor}. If $n=2$, $A_2$ will be called a \emph{complementer factor} to $A_1$. We have the following obvious result:

\begin{lemma}
	Let $G$ be an abelian group and $B$ be a subgroup of $G$. Then $G=A+B$ is a group factorization if and only if $A$ is a set of representatives of cosets of $G/B$.
\end{lemma}

For any subset $B$ of an abelian group $(G,+)$, an element $g\in G$ such that $g+B=B$ is called a \emph{period} of $B$.
Let $\pi(B)$ be the set of all periods of $B$, one can verify that $\pi(B)$ is a subgroup of $G$. We call it the \emph{stable subgroup} of $B$. For every element $b\in B$, $b+\pi(B)\subset B$, so $B$ is the union (and in fact, disjoint union) of some cosets of $\pi(B)$ in $G$.

The $n$-th cyclotomic polynomial $\Phi_n(x)$
is defined to be the monic minimal polynomial of an $n$-th primitive root of unity over $\mathbb{Q}$, and hence is irreducible in $\mathbb{Q}[x]$. In fact, $\Phi_n(x)\in\mathbb{Z}[x]$. For $n\ne m$, $\gcd(\Phi_n(x),\Phi_m(x))=1$ since they have no common zeros in $\mathbb{C}$.

\begin{lemma}[{\cite[\S 46]{nagell2021introduction}}] We have
	$$\Phi_n(1)=\begin{cases}p,&\text{if $n$ is a power of a prime $p$,}\\1,&\text{otherwise.}\end{cases}$$
\end{lemma}

With a subset of nonnegative integers $A$, we associate a \emph{mask polynomial}:
$$
	f_A(x)=\sum_{i\in A} x^i.
$$
The factorization of mask polynomials is closely connected to the factorization of sets. One can verify that
\begin{lemma}
	\label{mask-factor}
	Let $A,B$ be two finite sets of nonnegative integers, and $C=A+B$. Then $$
		f_A(x)f_B(x)=f_C(x).$$
\end{lemma}
\section{Group Factorizations}
In this section, we prove some results on group factorizations. Using these results, some necessary and sufficient conditions for the existence of certain splitter sets will be established.
\begin{lemma}
	Let $p$ be a prime. Suppose that $A$ is a finite set of nonnegative integers and $f_A(x)=\sum_{i\in A} x^i$ is its mask polynomial. If $\Phi_{p^k}(x)\mid f_A(x)$, then $p\mid|A|$ and $A\pmod {p^k}$ as a multiset must be the union of $|A|/p$ arithmetic progressions with difference $p^{k-1}$ and length $p$.
	\label{lem-ap}
\end{lemma}
\begin{proof}
	In fact,
	$$\Phi_{p^k}(x)=\Phi_p(x^{p^{k-1}})=1+x^{p^{k-1}}+\cdots+x^{(p-1)\cdot p^{k-1}}.$$
	Since $\Phi_{p^k}(1)=p$, $\Phi_{p^k}(x)\mid f_A(x)$ implies that $p\mid f_A(1)=|A|$.

	Let $$\overline{f_A}(x)=\sum_{i\in A}x^{i\bmod p^k},$$ where
	$i\mod p^k\in[0,p^k-1]$, then $$f_A(x)\equiv \overline{f_A}(x)\pmod{(x^{p^k}-1)}.$$ Since $\Phi_{p^k}(x)\mid (x^{p^k}-1)$,  then $\Phi_{p^k}(x)\mid \overline{f_A}(x)$. In \cite[p374]{debruijn1953factorization}, it was mentioned that if $f(x)$ is a polynomial with nonnegative integral coefficients, $(1+x^b+x^{2b}+\cdots+x^{ab})\mid f(x)$ and the degree of $f(x)$ is smaller than $(a+1)b$, then $f(x)/(1+x^b+\cdots+x^{ab})$ is a polynomial with nonnegative integral coefficients. We set $b=p^{k-1}, a=p-1$ and note that $\deg f(x)\le p^k-1<(a+1)b$, then we can apply this argument to see that $g(x)=\overline{f_A}(x)/\Phi_{p^k}(x)$ is a polynomial with nonnegative integral coefficients. We suppose
	$$
		g(x)=x^{i_1}+x^{i_2}+\cdots+x^{i_m},
	$$
	where $i_1,i_2,\ldots,i_m$ do not need to be distinct.
	Since $$g(1)=\overline{f_A}(1)/\Phi_{p^k}(1)=|A|/p,$$ then $m=|A|/p$. From
	$$
		\overline{f_A}(x)=g(x)\Phi_{p^k}(x)=\sum_{j=1}^m x^{i_j}(1+x^{p^{k-1}}+\cdots+x^{(p-1)p^{k-1}}),
	$$
	we see that $A\pmod{p^k}$ as a multiset is the union of $A_j$ for $j=1,2,\ldots,m$, where $$\begin{aligned}A_j&=\{i_j\}+p^{k-1}\cdot[0,p-1]\\&=\{i_j\}+\{0,p^{k-1},2p^{k-1},\ldots,(p-1)p^{k-1}\}.\end{aligned}$$  Clearly each $A_j$ is an arithmetic progression with difference $p^{k-1}$ and length $p$.
\end{proof}
\begin{lemma}
	\label{lem-partition}
	Let $p$ be a prime, $N=p^am$, where $p\nmid m, a\ge 1$. Suppose that $G=A+B$ is a factorization of $G=\mathbb{Z}_N$. Let $f_A, f_B$ be the mask polynomials of $A, B$ respectively.
	For any subset $X$ of $G$, we define $$\mathfrak{M}_X=\{1\le j\le a:\Phi_{p^j}(x)\mid f_X(x)\},$$ then $|\mathfrak{M}_A|=v_p(|A|), |\mathfrak{M}_B|=v_p(|B|)$ and $\{\mathfrak{M}_A,\mathfrak{M}_B\}$ is a partition of
	$\{1,2,\ldots,a\}$. \label{lemma-divisors}
\end{lemma}
\begin{proof}
	By Lemma \ref{mask-factor} there is a factorization
	$$
		f_A(x)f_B(x)\equiv 1+x+x^2+\cdots+x^{N-1}\pmod{(x^N-1)}.
	$$
	Since for every $1\le j\le a$, $\Phi_{p^j}(x)\mid (1+x+\cdots+x^{N-1})$, hence for every $1\le j\le a$,
	$\Phi_{p^j}(x)$ must divide $f_A(x)$ or $f_B(x)$, that is,
	$$\mathfrak{M}_A\cup\mathfrak{M}_B=\{1,2,\ldots,a\}.$$
	But notice that $\Phi_{p^j}(1)=p$, so
	$|\mathfrak{M}_A|\le v_p(f_A(1))=v_p(|A|)$ and similarly
	$|\mathfrak{M}_B|\le v_p(f_B(1))=v_p(|B|)$.
	Hence
	$$
		|\mathfrak{M}_A|+|\mathfrak{M}_B|\le v_p(|A|)+v_p(|B|)=a.
	$$
	So $|\mathfrak{M}_A|=v_p(|A|), |\mathfrak{M}_B|=v_p(|B|)$ and $\{\mathfrak{M}_A,\mathfrak{M}_B\}$ is a partition of
	$\{1,2,\ldots,a\}$.
\end{proof}
\begin{thm}
	\label{thm-dir-fct}
	Let $p$ be a prime, $G=\mathbb{Z}_{p^am}$, where $p\nmid m$, $A$ be a subset of $G$ of size $p^n, 1\le n\le a$, then $A$ is a direct factor of $G$ if and only if there exist integers
	$1\le i_1<i_2<\ldots<i_n\le a$, such that we can label elements of $A$ as
	$a_{b_1b_2\cdots b_n}, b_i\in\{0,1,\ldots,p-1\}$, where:
	\begin{enumerate}[label=({\arabic*})]
		\item the $i_j$-th to last digit of $a_{b_1b_2\cdots b_n}$ in base $p$ is $b_j$, that is,
		      $\lfloor a_{b_1b_2\cdots b_n}/p^{i_j-1}\rfloor \equiv b_j\pmod{p}$;
		\item $a_{b_1b_2\cdots b_{j-1} b_j *\cdots *}\equiv a_{b_1b_2\cdots b_{j-1} 0 *\cdots *}+b_jp^{i_j-1}\pmod{p^{i_j}}$,
		      where $*$ stands for arbitrary values in $[0,p-1]$ (not necessary identical for corresponding coordinates).
	\end{enumerate}
\end{thm}
\begin{proof}
	First we prove, by induction on $n$, that:
	\begin{quotation}(*) if $A$ is a subset of $G$ of size $p^n$ and $\Phi_{p^i}(x)\mid f_A(x)$ for $i\in\{i_1,i_2,\ldots,i_n\},
			1\le i_1<i_2<\cdots<i_n\le a$, then $A$ must satisfy the condition in the theorem.\end{quotation}

	Then combined with Lemma \ref{lem-partition} we can prove the ``only if'' part of the theorem.

	The case $n=1$ is obvious: by Lemma \ref{lem-ap} $A$ will be an arithmetic progression with difference $p^{i_1-1}$ and length $p$. We may choose $a_0\in A$ such that the $i_1$-th to last digit of $a_0$ is zero, $a_i$ to be the unique element of $A$ such that $a_i\equiv a_0+ip^{i_1-1}\pmod{p^{i_1}}$ for $i\in [1,p-1]$.

	Suppose that the (*) holds for $n=\nu-1$, consider the case $n=\nu$.
	Since $\Phi_{p^{i_\nu}}(x)\mid f_A(x)$, $A\pmod{p^{i_\nu}}$ must be the union of
	$p^{\nu-1}$ arithmetic progressions with difference
	$p^{i_\nu-1}$ and length $p$:
	$$
		A\pmod{p^{i_\nu}}=\bigcup_{i=1}^{p^{\nu-1}} A_i,
		\quad A_i=\{a_i\}+p^{i_\nu-1}\cdot [0,p-1].
	$$
	Without loss of generality, we suppose that the
	$i_\nu$-th to last digit under base $p$ of $a_i$ is $0$.
	Then elements of $A_i$ share the same last $(i_\nu-1)$ digits and differ at the $i_\nu$-th to last digit.

	We turn to $A'=\{a_i:i\in [1,p^{\nu-1}]\}$, then as multisets $$A\equiv\bigcup_{i=1}^{p} A'\pmod{p^{i_\nu-1}}$$ and hence $$f_A(x)\equiv pf_{A'}(x)\pmod{(x^{p^{i_\nu-1}}-1)}.$$ We prove that $A'$ satisfies the induction hypothesis. In fact, $A'$ is a subset of $G$ of size $p^{\nu-1}$, and
	if $\Phi_{p^j}(x)\mid f_A(x)$ for some $1\le j<i_\nu$, then $\Phi_{p^j}(x)\mid pf_{A'}(x)$ since $\Phi_{p^j}(x)\mid (x^{p^j}-1)\mid (x^{p^{i_\nu-1}}-1)$, which implies that $\Phi_{p^j}(x)\mid f_{A'}(x)$.
	So $\Phi_{p^i}(x)\mid f_{A'}(x)$ for $i\in\{i_1,\ldots,i_{\nu-1}\}$.
	By induction hypothesis, we can label the elements of $A'$
	as
	$a_{b_1b_2\cdots b_{\nu-1}}, b_i\in\{0,1,\ldots,p-1\}$, such that the $i_j$-th to last digit of $a_{b_1b_2\cdots b_{\nu-1}}$ is $b_j$, and
	$$a_{b_1b_2\cdots b_{j-1} b_j *\cdots *}\equiv a_{b_1b_2\cdots b_{j-1} 0 *\cdots *}+b_jp^{i_j-1}\pmod{p^{i_j}},$$
	where $*$ stands for arbitrary values in $[0,p-1]$.
	Then let $a_{b_1b_2\cdots b_{\nu-1}0}:=a_{b_1\ldots b_{\nu-1}}$, so
	$$
		A\pmod{p^{i_\nu}}=\bigcup_{i=1}^{p^{\nu-1}}\{a_{b_1b_2\cdots b_{\nu-1}0}\}+p^{i_\nu-1}\cdot [0,p-1].
	$$
	We let $a_{b_1b_2\cdots b_{\nu-1}b_{\nu}} (b_{\nu}>0)$ be the element of $A$ that corresponds to
	$a_{b_1b_2\cdots b_{\nu-1}0}+p^{i_\nu-1}b_{\nu}$ in the decomposition above, and (*) is then proved.

	Now we prove the ``if'' part of the theorem. If $A$ satisfies the condition in the theorem, we let \[
		\begin{aligned}
			C_1   & =\{0,1,\ldots,p^{i_1-1}-1\}
			,                                                        \\ C_2&=\{0,p^{i_1},2p^{i_1},\ldots,p^{i_2-1}-p^{i_1}\},\\&\vdots \\
			C_{n} & =\{0,p^{i_{n-1}},\ldots,p^{i_n-1}-p^{i_{n-1}}\}, \\ C_{n+1}&=\{0,p^{i_n},\cdots,|G|-p^{i_n}\},
		\end{aligned}\] then \[\begin{aligned}|C_1| & =p^{i_1-1},         \\|C_2|&=p^{i_2-1-i_1}, \\&\vdots\\
               |C_n| & =p^{i_n-1-i_{n-1}}, \\|C_{n+1}|&=|G|/p^{i_n}.\end{aligned}\]
	We prove that $G=A+C_1+C_2+\cdots+C_n+C_{n+1}$ is a factorization. Notice that $|A||C_1||C_2|\cdots |C_{n+1}|=p^{n+i_1-1+i_2-i_1-1+\cdots+i_n-i_{n-1}-1-i_n}|G|=|G|$, so we only have to prove that every element in $A+C_1+C_2+\cdots+C_{n+1}$ can be expressed as a sum in $A+C_1+C_2+\cdots+C_{n+1}$ in only one way. Suppose there exist
	$a_{b_1b_2\cdots b_n}, a_{\beta_1,\beta_2\cdots\beta_n}\in A,
		c_i,\sigma_i\in C_i$ such that
	$$\begin{aligned}
			       & a_{b_1b_2\cdots b_n}+c_1+c_2+\cdots+c_n+c_{n+1}                                    \\
			\equiv & a_{\beta_1,\beta_2\cdots\beta_n}+\sigma_1+\cdots+\sigma_n+\sigma_{n+1}\pmod{p^am}.\end{aligned}$$
	First, taking congruences modulo $p^{i_1-1}$, we notice that for $i>1$,
	$c_i,\sigma_i\equiv0\pmod{p^{i_1}}$, then we get
	$$a_{b_1b_2\cdots b_n}+c_1\equiv a_{\beta_1,\beta_2\cdots\beta_n}+\sigma_1\pmod{p^{i_1-1}}.
	$$
	Since $a_{b_1*\cdots*}\equiv a_{0*\cdots*}+b_1p^{i_1-1}\pmod{p^{i_1}}$, then $a_{b_1*\cdots*}\equiv a_{0*\cdots*}\equiv 0\pmod{p^{i_1-1}}$. Analogously $a_{\beta_1*\cdots*}\equiv 0\pmod{p^{i_1-1}}$. Hence $c_1=\sigma_1$. Then
	$$a_{b_1b_2\cdots b_n}\equiv a_{\beta_1\beta_2\cdots\beta_n}\pmod{p^{i_1}}
	$$
	and hence $b_1=\beta_1$. In the same manner we can prove that
	$c_2=\sigma_2$ and $b_2=\beta_2$, $c_3=\sigma_3$,\ldots,$c_n=\sigma_n$ and we get
	$b_1\cdots b_n=\beta_1\cdots \beta_n$.
	Thus $c_{n+1}\equiv \sigma_{n+1}\pmod{p^am}$, which implies that
	$c_{n+1}=\sigma_{n+1}$. If we set $B=C_1+C_2+\cdots+C_n+C_{n+1}$, then $A+B$ is a factorization.
\end{proof}
Notice that, when the condition in the theorem is satisfied, all elements in $A$ modulo $p^{i_n}$ are distinct.
\begin{remark}
	In \textnormal{\cite[Problem 77]{fuchs1960abelian}}, Fuchs posed the question of whether, in every group factorization $G=S+T$ of a finite abelian group $G$, one of the sets $S$ or $T$ can be replaced by a subgroup of $G$. All cyclic groups satisfying this property have been characterized \textnormal{(\cite{sands2000replacement})}. Our theorem demonstrates that if $G=\mathbb{Z}_{p^km}$ and $|S|$ is a power of prime, then if $i_1=1,i_2=2,\ldots,i_n=n$, then $T$ can be replaced by a cyclic subgroup of $G$. This also explains why a counterexample proposed by Sands \textnormal{(\cite[Example 1]{sands1962problem})} works: it corresponds to the case where $k=\alpha, m=1, n=\alpha-2, i_1=2, i_2=3,\ldots,i_n=\alpha-1$.
\end{remark}
\begin{cor}[{\cite[Theorem 8]{zhang2018nonexistence}}]
	Let $q,k$ be odd primes, $g$ be a primitive root modulo $q$ and
	$\mu=\gcd(q-1,\{\ind_g(j):j\in[1,k]^*\})$. Then there exists a perfect
	$B[0,k](q)$ set if and only if
	$q\equiv 1\pmod{\mu k}$ and $$\left|\{\ind_g(j)/\mu \pmod{k}:j\in [1,k]\}\right|=k.$$
\end{cor}
\begin{proof}
	Since $k$ is an odd prime and $\ind_g(1)=0$, the condition in Theorem \ref{thm-dir-fct} can be described as: there exists a positive integer $a\le v_k(q-1)$ such that elements of $A=\{\ind_g(j):j\in[1,k]^*\}$ modulo $k^a$ are exactly
	$$0,k^{a-1},2k^{a-1},\ldots,(k-1)k^{a-1}.$$ If that is the case, then
	$v_k(\mu)=a-1<v_k(q-1)$, hence $\mu k \mid (q-1)$, and since
	$\gcd(\mu/k^{a-1},k)=1$, $$\begin{aligned}  & |\{\ind_g(j)/\mu \pmod{k}:j\in [1,k]\}| \\=&|\{(\ind_g(j)/k^{a-1})(\mu/k^{a-1})^{-1} \pmod{k}:j\in [1,k]\}|
                \\=&|\{0,1,\ldots,k-1\}|=k.\end{aligned}$$ On the other hand, if $\mu k\mid (q-1)$ and $|\{\ind_g(j)/\mu \pmod{k}:j\in [1,k]\}|=k$, then $v_k(\mu)<v_k(q-1)$. Let $a=v_k(\mu)+1$, then $\gcd(\mu/k^{a-1},k)=1$, $$\begin{aligned}&|\{\ind_g(j)/k^{a-1} \pmod{k}:j\in [1,k]\}|\\=&|\{(\ind_g(j)/\mu)(\mu/k^{a-1}) \pmod{k}:j\in [1,k]\}|=k,\end{aligned}$$ which implies that $$\{\ind_g(j)/k^{a-1} \pmod{k}:j\in [1,k]\}=\{0,1,\ldots,k-1\},$$ hence $A$ satisfies the condition in Theorem \ref{thm-dir-fct}.
\end{proof}
\begin{exampl}[{\cite[Example 10.6]{szabo2009factoring}}]
	Let us investigate if $M=\{1,2,3,4,5\}$ splits $\mathbb{Z}_{421}^\times$. We select $2$ as a primitive root modulo $421$, then we are left with the problem of whether $A=\{0,1,404,2,278\}$ has a complementer factor in $\mathbb{Z}_{420}$. Since $404\equiv 4\pmod5, 278\equiv 3\pmod 5$, we deduce that $A$ is a direct factor of $\mathbb{Z}_{420}$ and $M$ does split $\mathbb{Z}_{421}^\times$.
\end{exampl}
\begin{exampl}[{\cite[Example 10.5]{szabo2009factoring}}]
	Let us decide if $M=\{1,2,3\}$ splits $\mathbb{Z}_{103}^\times$, or equivalently, by noticing that $5$ is a primitive root modulo $103$, if $A=\{0,44,39\}$ has a complementer factor in $\mathbb{Z}_{102}$. Since $3\mid 39, 3\nmid 44$, so we deduce that $A$ has no complementer factor in $\mathbb{Z}_{102}$ and $M$ does not split $\mathbb{Z}_{103}^\times$.
\end{exampl}
\begin{cor}
	Let $G=\mathbb{Z}_{2^am}$, $2\nmid m$, $A$ be a subset of $G$ of size $2^n$, then $A$ is a direct factor of $G$ if and only if there exist
	$1\le i_1<i_2<\ldots<i_n\le a$, such that we can label elements of $A$ as
	$a_{b_1b_2\cdots b_n}, b_i\in\{0,1\}$, the $i_j$-th to last digit of $a_{b_1b_2\cdots b_n}$ is $b_j$, and
	$a_{b_1b_2\cdots b_{j-1} 1 *\cdots *}\equiv a_{b_1b_2\cdots b_{j-1} 0 *\cdots *}+2^{i_j-1}\pmod{2^{i_j}}$,
	where $*$ stands for arbitrary values in $[0,1]$.
	\label{cor-2n}
\end{cor}
We will use this corollary in the following corollaries and in Section IV. For simplicity, we will write $i=i_1,j=i_2,k=i_3$.
\begin{cor}[{\cite[Theorem 2]{klove2011some}}] Let $q$ be an odd prime, then a nonsingular perfect $B[0,2](q)$ set $B$ exists if and only if $\ord_q(2)$ is even.
\end{cor}
\begin{proof}
	Let $q=2^am+1$, where $a\ge 1, 2\nmid m$. A nonsingular perfect $B[0,2](q)$ set $B$ exists if and only if $\ind_g(2)\equiv 2^{i-1}\pmod{2^i}$ for some $1\le i\le a$. It is equivalent to $2^a\nmid \ind_g(2)$, which is in turn equivalent to saying that $\ord_q(2)$ is even.\end{proof}
\begin{cor}[{\cite[Corollary 3]{klove2012codes}}] Let $q\equiv 1\pmod4$ be an odd prime, then a nonsingular perfect $B[-2,2](q)$ set $B$ exists if and only if $v_2(\ord_q(2))\ge 2$.
	\label{-2-2}
\end{cor}\begin{proof}
	Let $q=2^am+1$, where $a\ge 2, 2\nmid m$. Then a nonsingular perfect $B[-2,2](q)$ set $B$ exists if and only if $\{0, \ind_g(2), \ind_g(-1),\ind_g(-2)\}=\{a_{00},a_{01},a_{10},a_{11}\}$, where
	$2^j\mid a_{00}, 2^{i-1} \mid a_{1*}, 2^i \nmid a_{1*}, 2^{j-1}\mid a_{01}, 2^j\nmid a_{01}, 2^j\mid(a_{11}-a_{10}+2^{j-1})$
	and the $j$-th to last binary digit of $a_{10},a_{11}$ are $0,1$ respectively
	for some $1\le i<j\le a$. If it is true, then $a_{00}$ must be 0, and since $\ind_g(-1)=\frac{q-1}2=2^{a-1}m$,
	$a_{01}$ must be $\ind_g(-1)$ and $j=a$. Therefore we only require that $\{\ind_g(2),\ind_g(-2)\}=\{a_{10},a_{11}\}$. Since $\ind_g(-2)
		=\ind_g(2)\pm \frac{q-1}2$, we always have $a_{10}\equiv a_{11}+2^{a-1}\pmod{2^a}$. So the condition is equivalent to saying that
	$2^{i-1}\mid \ind_g(2), 2^i\nmid \ind_g(2)$ for some $1\le i\le a-1$, which in turn is equivalent to
	$2^{a-1}\nmid \ind_g(2)$, or $v_2(\ord_q(2))\ge 2$.
\end{proof}
\begin{cor}[{\cite[Theorems 4.4--4.5]{yuan2020existence}}] Let $q\equiv1\pmod4$ be an odd prime, then a nonsingular perfect $B[-1,3](q)$ set $B$ exists if and only  if $4\mid \ord_q(2)$ and $\ord_q(-\frac32)$ is odd.
\end{cor}
\begin{proof}
	Let $q=2^am+1$, where $a\ge 2, 2\nmid m$. Then a nonsingular perfect $B[-1,3](q)$ set $B$ exists if and only
	if $\{0, \ind_g(2), \ind_g(-1),\ind_g(3)\}=\{a_{00},a_{01},a_{10},a_{11}\}$ and $a_{**}$'s should be as described in the proof to Corollary \ref{-2-2}. As before, it is equivalent to
	$\{\ind_g(2),\ind_g(3)\}=\{a_{10},a_{11}\}$, which in turn is equivalent to $v_2(\ind_g(2))=v_2(\ind_g(3))<a-1$
	and $v_2(-\ind_g(-1)-\ind_g(2)+\ind_g(3))\ge a$. Since $2^a\mid (q-1)$, $-\ind_g(-1)-\ind_g(2)+\ind_g(3)\equiv
		\ind_g(-\frac{3}2)\pmod{2^a}$. So $v_2(-\ind_g(-1)-\ind_g(2)+\ind_g(3))\ge a$ if and only if $\ord_q(-\frac32)$ is odd.
	If it is the case, then $v_2(\ind_g(3))=v_2(\ind_g(2))$ always holds. Moreover, $v_2(\ord_q(2))<a-1$ if and only if $v_2(\ord_q(2))>1$, or equivalently $4\mid\ord_q(2)$. So the necessary and sufficient condition for the existence of a perfect $B[-1,3]^*(q)$ set is that $4\mid \ord_q(2)$ and $\ord_q(-\frac32)$ is odd.
\end{proof}
\begin{remark}
	In \textnormal{\cite{yuan2020existence}}, the case of $B[-1,3](q)$ set was divided into two cases: $q\equiv 1\pmod8$ and $q\equiv 5\pmod 8$. The existence condition for perfect $B[-1,3](q)$ sets was stated as:
	(i)  $\ord_q(-\frac32)$ is odd and $4\mid\ord_q(2)$ if
	$q\equiv 1\pmod8$ (which is its Theorem 4.5), (ii) $6$ is a quartic residue modulo $q$ if $q\equiv 5\pmod 8$ (its Theorem 4.4).

	Here we show that the conditions in these two theorems are equivalent, that is, when $q\equiv 5\pmod 8$, $\ord_q(-\frac32)$ is odd and $4\mid\ord_q(2)$ if and only if $6$ is a quartic residue modulo $q$, and hence we do not need a different condition for the case $q\equiv 5\pmod 8$. In fact, when $q\equiv 5\pmod 8$, $2$ is not a quadratic residue modulo $q$ (\textnormal{{\cite[Theorem 81]{nagell2021introduction}}}), so $[\mathbb{Z}_q^\times:\langle2\rangle]$ is odd, and $$\ord_q(2)=|\langle2\rangle|\equiv |\langle2\rangle|[\mathbb{Z}_q^\times:\langle2\rangle]=|\mathbb{Z}_q^\times|\equiv0\pmod4,$$ thus $4\mid\ord_q(2)$ always holds. Since $4\mid(q-1)$ but $8\nmid(q-1)$, $\ord_q(-\frac32)$ is odd if and only if
	$4\mid [\mathbb{Z}_q^\times:\langle-\frac32\rangle]$, which means that $-\frac32$ is a quartic residue. However when $q\equiv 5\pmod 8$, $(-4)^{\frac{q-1}4}=(-1)^{\frac{q-1}4}2^{\frac{q-1}2}=-2^{\frac{q-1}2}\equiv 1\pmod q$, so (by, for example, \textnormal{\cite[Theorem 69]{nagell2021introduction}}) $-4$ is a quartic residue. Hence $-\frac32$ is a quartic residue if and only if $6=(-4)(-\frac32)$ is, and we see that (i) and (ii) are equivalent when $q\equiv 5\pmod 8$.
\end{remark}
\begin{thm}
	Suppose $k$ is a power of an odd prime and
	$q\equiv 1\pmod{2k}$ is a prime, then there exists a nonsingular perfect $B[-k,k](q)$ set if and only if
	$\overline{A}=\{\ind_g(i)\pmod{\frac{q-1}2}:i\in[1,k]\}$ is a direct factor of $\mathbb{Z}_{\frac{q-1}2}$.
	\label{cor-2k}
\end{thm}
\begin{proof}
	If $B$ is a perfect $B[-k,k](q)$ set, then $$\mathbb{Z}_q^\times=[1,k]\cdot\{\pm1\}\cdot B.$$
	Let $g$ be a primitive root modulo $q$. Let $A=\{\ind_g(i):i\in[1,k]\}, C=\{\ind_g(i):i\in B\}$, then
	$$\mathbb{Z}_{q-1}=A+\left\{0,\frac{q-1}2\right\}+C$$ will be a factorization. Since $\{0,\frac{q-1}2\}$ is a subgroup of $\mathbb{Z}_{q-1}$, then $A+C$ is the complete residue system modulo $\frac{q-1}2$.
	That is, $\mathbb{Z}_{\frac{q-1}2}=\overline{A}+\overline{C}$ is a factorization, where
	$\overline{A}=\left\{a\pmod{\frac{q-1}2}:a\in A\right\},\overline{C}=\left\{c\pmod{\frac{q-1}2}:c\in C\right\}$. On the other hand, suppose that $\mathbb{Z}_{\frac{q-1}2}=\overline{A}+\overline{C}$ is a factorization. If $x,y\in A+C$ such that $x\equiv y+\frac{q-1}2\pmod{(q-1)}$, then $x\equiv y\pmod{\frac{q-1}2}$ and hence $x=y$. So $\mathbb{Z}_{q-1}=A+C+\{0,\frac{q-1}2\}$ is a factorization, and
	$B=\{g^i:i\in C\}$ is a perfect
	$B[-k,k](q)$ set.
\end{proof}
Whether $\overline{A}$ is a direct factor of $\mathbb{Z}_{\frac{q-1}2}$ can be decided by Theorem \ref{thm-dir-fct}. We get an alternative proof of {\cite[Theorem 7]{zhang2018nonexistence}}.
\begin{cor}
	Let $q,k$ be odd primes, $g$ be a primitive root modulo $q$ and
	$\mu=\gcd(\frac{q-1}2,\{\ind_g(j):j\in[-1,k]^*\})$. Then there exists a perfect
	$B[-k,k](q)$ set if and only if
	$q\equiv 1\pmod{2\mu k}$ and $\left|\{\ind_g(j)/\mu \pmod{k}:j\in [1,k]\}\right|=k$.
	\label{-k-k}
\end{cor}
\begin{proof}
	By Theorem \ref{cor-2k}, there exists a perfect
	$B[-k,k](q)$ set if and only if $\overline{A}=\{\ind_g(i)\bmod{\frac{q-1}2}:i\in[1,k]\}$ satisfies the condition in Theorem \ref{thm-dir-fct}. Since $k$ is an odd prime and $\ind_g(1)=0\in \overline{A}$, the condition can be described as: there exists a positive integer $a\le v_k(q-1)$ such that elements of $\overline{A}$ modulo $k^a$ are exactly
	$$0,k^{a-1},2k^{a-1},\ldots,(k-1)k^{a-1}.$$ If that is the case, then
	$$v_k(\mu)=a-1<v_k(q-1)=v_k\left(\frac{q-1}2\right),$$ $\mu k \mid \frac{q-1}2$ and hence $2\mu k\mid (q-1)$, and since
	$\gcd(\mu/k^{a-1},k)=1$,
	\[\begin{aligned}  & |\{\ind_g(j)/\mu \pmod{k}:j\in [1,k]\}| \\=&|\{(\ind_g(j)/k^{a-1})(\mu/k^{a-1})^{-1} \pmod{k}:j\in [1,k]\}|
                \\=&|\{0,1,\ldots,k-1\}|=k.\end{aligned}\] On the other hand, if $2\mu k\mid (q-1)$ and $|\{\ind_g(j)/\mu \pmod{k}:j\in [1,k]\}|=k$, then $v_k(\mu)<v_k(q-1)$. Let $a=v_k(\mu)+1$, then $\gcd(\mu/k^{a-1},k)=1$, $$
		\begin{aligned}
			&|\{\ind_g(j)/k^{a-1} \pmod{k}:j\in [1,k]\}|\\=&|\{(\ind_g(j)/\mu)(\mu/k^{a-1}) \pmod{k}:j\in [1,k]\}|=k,\end{aligned}$$ which implies that $$\{\ind_g(j)/k^{a-1} \pmod{k}:j\in [1,k]\}=\{0,1,\ldots,k-1\},$$ hence $\overline{A}=\{\ind_g(i)\pmod{\frac{q-1}2}:i\in[1,k]\}$ satisfies the condition in Theorem \ref{thm-dir-fct}.
\end{proof}
By Corollary \ref{-k-k}, the existence condition of ${B[-3,3](q)}$ sets can be established.
\begin{cor}[Existence condition of ${B[-3,3](q)}$ sets]
	Let $q\equiv 1\pmod 6$ be a prime. Then a perfect $B[-3,3](q)$ set exists if and only if $2\notin\langle 6,8\rangle\subset\mathbb{Z}_q^\times$.
\end{cor}
\begin{proof}Let $q=3^am+1,3\nmid m$.
	By Corollary \ref{cor-2k}, a perfect $B[-3,3](q)$ set exists if and only if $\{0,\ind_g(2),\ind_g(3)\}$ is a factor of $\mathbb{Z}_{\frac{q-1}2}$.
	By Theorem \ref{thm-dir-fct}, the necessary and  sufficient condition is  $v_3(\ind_g(2))=v_3(\ind_g(3))=v_3(\ind_g(3)-\ind_g(2))<a$.

	We can prove that it is equivalent to $2\notin\langle 6,8\rangle$. We know from
	Lemma \ref{lem-in} that  $2\notin\langle 6,8\rangle$ is equivalent to
	$[\mathbb{Z}_q^\times:\langle 6,8\rangle]\nmid \ind_g(2)$. However
	\[[\mathbb{Z}_q^\times:\langle 6,8\rangle]=\gcd(\ind_g(6),\ind_g(8),q-1).\]
	Since $\ind_g(6)\equiv\ind_g(2)+\ind_g(3)\pmod{(q-1)}, \ind_g(8)\equiv 3\ind_g(2)\pmod{(p-1)}$, we have
	$$\begin{aligned}&\gcd(\ind_g(6),\ind_g(8),q-1)\\=&\gcd(\ind_g(2)+\ind_g(3),3\ind_g(2),q-1)=:d.\end{aligned}$$ Now we prove that $d\nmid \ind_g(2)$ is equivalent to our necessary and sufficient condition for the existence of $B[-3,3](q)$ sets. Suppose that $d\mid \ind_g(2)$, then $$\begin{aligned}d & =\gcd(\ind_g(2)+\ind_g(3),3\ind_g(2),q-1,\ind_g(2)) \\&=\gcd(\ind_g(2)+\ind_g(3),\ind_g(2),q-1)\\&
               =\gcd(\ind_g(3),\ind_g(2),q-1).\end{aligned}$$ If $v_3(\ind_g(2))=v_3(\ind_g(3))=v_3(\ind_g(3)-\ind_g(2))=:i<a$, then
	$$\begin{aligned}&\ind_g(2)+\ind_g(3)\\\equiv &\ind_g(3)-\ind_g(2)+\ind_g(2)\cdot 2\\\equiv & {}3\ind_g(2)\equiv 0 \pmod{3^{i+1}}.\end{aligned}$$
	Then $3^{i+1}\mid \ind_g(2)+\ind_g(3),3^{i+1}\mid \ind_g(2),3^{i+1}\mid (q-1)$. Hence $3^{i+1}\mid d\mid \ind_g(3)$, a contradiction. So $v_3(\ind_g(2))=v_3(\ind_g(3))=v_3(\ind_g(3)-\ind_g(2))<a$ cannot hold.

	Then suppose that $d\nmid \ind_g(2)$.
	Let $i:=v_3(d)$.
	If $v_3\left(\frac{3\ind_g(2)}d\right)\ge 1$, then $\frac{\ind_g(2)}d$ is an integer, and $d\mid \ind_g(2)$, a contradiction.
	So $v_3\left(\frac{3\ind_g(2)}d\right)=0$,
	and $i=v_3(d)=v_3(3\ind_g(2))=v_3(\ind_g(2))+1$
	. We have $i\le v_3(q-1)\le a$ because
	$d\mid (q-1)$. Since $d\mid (\ind_g(3)+\ind_g(2))$,
	$$\begin{aligned}&\ind_g(3)-\ind_g(2)\\\equiv& \ind_g(3)+\ind_g(2)-2\ind_g(2)\\\equiv& 3^{i-1}\pmod{3^i}.\end{aligned}$$
	Hence $v_3(\ind_g(3)=v_3(\ind_g(2))=v_3(\ind_g(3)-\ind_g(2))=i-1<i\le a$, which completes our proof.
\end{proof}
Next, a characterization of the splitter set $B$ is described. We prove that when $|A|$ is a power of prime, $B$ will have some periodicity:
\begin{thm}
	If $G=A+B$ is a factorization of a cyclic group, $|A|=p^n$ is a power of prime. Suppose $i_n$ is the maximal integer such that
	$\Phi_{p^{i_n}}(x)\mid f_A(x)$. Then $p^{i_n}\mid [G:\pi(B)]$, that is, any period of $B$ is divisible by $p^{i_n}$.
	\label{thm-prd}
\end{thm}
\begin{proof}
	We know that $B$ is the union of some cosets of $\pi(B)$ in $G$.
	Let $B'=\{b_1,\ldots,b_l\}\subset G$ be a set of representatives of cosets of $\pi(B)$ such that
	$$B=\bigsqcup_{i=1}^l (b_i+\pi(B)).$$ Then $B=B'+\pi(B)$,
	$G=A+B'+\pi(B)$, and $A+B'$ is the set of representatives of all cosets of $\pi(B)$ in $G$. So
	$$A+B'\equiv \{0,1,\ldots, [G:\pi(B)]-1\} \pmod{[G:\pi(B)]},$$
	that is, $\overline{A}:=A\pmod{[G:\pi(B)]}$ is a direct factor of
	$\mathbb{Z}_{[G:\pi(B)]}$.
	By Lemma \ref{lemma-divisors},
	$\mathfrak{M}_{\overline{A}}:=\{j:\Phi_{p^j}(x)\mid f_{\overline{A}}(x)\}\subset [1,b]$, where
	$b=v_p([G:\pi(B)])$. Since
	$$
		f_{\overline{A}}(x)\equiv f_A(x)\pmod{(x^{[G:\pi(B)]}-1)}
	$$ and for $1\le j\le b$,
	$\Phi_{p^j}(x)\mid (x^{[G:\pi(B)]}-1)$,  hence $\mathfrak{M}_{\overline{A}}\subset \mathfrak{M}_A$, and by noting that
	$|\mathfrak{M}_{\overline{A}}|=n=|\mathfrak{M}_A|$, we have $\mathfrak{M}_A=\mathfrak{M}_{\overline{A}}$.
	Therefore $i_n=\max \mathfrak{M}_A=\max\mathfrak{M}_{\overline{A}}\le b$, that is, $p^{i_n}\mid [G:\pi(B)]$.
\end{proof}
\begin{remark}
	We have already seen in the proof of Theorem \ref{thm-dir-fct} that we can
	construct a splitter set $B$ such that $[G:\pi(B)]=p^{i_n}$.
\end{remark}
\section{Existence of splitting with $|A|=8$}
As an application of Corollary \ref{cor-2n}, let's deal with the case $|A|=8$, i. e. $n=3$. Remember that we write $i=i_1,j=i_2,k=i_3$. Notice that since $i<j<k\le n$, $i\le n-2, j\le n-1$.
The $2$-adic valuations of elements in $\{a_{100},a_{101},a_{110},a_{111}\}$ are $(i-1)$, those of elements in $\{a_{010},a_{011}\}$ are $(j-1)$, and $v_2(a_{001})=k-1$.

First, we reprove the existence result for $B[-4,4](q)$ splitter sets:
\begin{thm}[{\cite[Theorem III.3]{ye2020some}}]
	Let $q$ be an odd prime, $q\equiv1\pmod 8$. Nonsingular perfect $B[-4,4](q)$ splitter sets exist if and only if  $\pm 4 \notin\langle 6,16\rangle$.
\end{thm}
\begin{proof}
	Suppose that $q=2^nm+1$ satisfies the condition, where $2\nmid m$, $n\ge 3$, and fix a primitive root $g$, then
	$\ind_g(-1)=2^{n-1}m\equiv 2^{n-1}\pmod{2^n}$. Hence
	$k=n,a_{000}=\ind_g(1)=0,a_{001}=\ind_g(-1)$.
	Since $v_2(a_{1**})=i-1, v_2(a_{01*})=j-1\le n-2$,
	thus $$\begin{aligned}
			v_2(\ind_g(\pm2)) & \le n-2, \\
			v_2(\ind_g(\pm3)) & \le n-2, \\
			v_2(\ind_g(\pm4)) & \le n-2.
		\end{aligned}$$
	We know that $$
		\ind_g(4)=\begin{cases}2\ind_g(2),&2\ind_g(2)<2^{n}m,\\2\ind_g(2)-2^{n}m,&2\ind_g(2)\ge2^{n}m,\end{cases}$$ hence $v_2(\ind_g(4))\ge v_2(\ind_g(2))+1$ always holds.
	Since
	$$
		\ind_g(-2)\equiv\ind_g(2)+\ind_g(-1)\pmod{(q-1)},$$
	we have
	$$
		\ind_g(-2)\equiv\ind_g(2)+2^{n-1}\pmod{2^n};
	$$
	and because $v_2(\ind_g(-1))=n-1>v_2(\ind_g(2))$, we get $v_2(\ind_g(2))=v_2(\ind_g(-2))=i-1\le n-3$, $v_2(\ind_g(4))=j-1\le n-2$, that is, $\ind_g(4)\equiv 2^{j-1}=a_{000}+2^{j-1}\pmod{2^{j}}$.
	Therefore we always have $$\{a_{010},a_{011}\}=\{\ind_g(4),\ind_g(-4)\}.$$ And $i=n-2$. So $\{\ind_g(2),\ind_g(3),\ind_g(-2),\ind_g(-3)\}=\{a_{100},a_{101},a_{110},a_{111}\}$.
	Again from
	$$
		\ind_g(-2)\equiv\ind_g(2)+2^{n-1}\pmod{2^n}
	$$
	we have $\{\ind_g(2),\ind_g(-2)\}=\{a_{100},a_{101}\}$ or $\{a_{110},a_{111}\}$, and $$
		\begin{aligned}
			\{\ind_g(3),\ind_g(-3)\}= & \{a_{100},a_{101},a_{110},a_{111}\} \\&\setminus\{\ind_g(2),\ind_g(-2)\}.
		\end{aligned}$$ No matter what the circumstances, we always have $\ind_g(3)\equiv \ind_g(2)+2^{j-1}\equiv \ind_g(2)+\ind_g(4)\pmod{2^{j}}$, which is equivalent to $v_2(\ind_g(3)-\ind_g(2))=v_2(\ind_g(4))$ since we know that $v_2(\ind_g(4))=j-1$.

	Hence the necessary condition is: $v_2(\ind_g(4))<n-1$,  $v_2(\ind_g(2))=v_2(\ind_g(3))$ and $v_2(\ind_g(3)-\ind_g(2))=v_2(\ind_g(4))$. It is easy to verify that it is also the sufficient condition.

	We can prove that this condition is equivalent to $\pm 4 \notin\langle 6,16\rangle$.
	Suppose $\pm 4 \notin\langle 6,16\rangle$. Denote $\ind_g(2)$ as $\alpha$, $\ind_g(3)$ as $\beta$, and let $d=\gcd(\ind_g(6), \ind_g(16), 2^n m)=\gcd(\alpha+\beta, 4\alpha, 2^n m)$.
	Then $$
		\begin{aligned}
			\ind_g(4) & \equiv2\alpha & \pmod{2^nm}, \\ \ind_g(-4)&\equiv 2\alpha+2^{n-1}m&\pmod{2^nm},
		\end{aligned}$$ and $\pm 4 \notin\langle 6,16\rangle$ is the same thing as the assertion that $d\nmid 2\alpha$ and $d\nmid 2\alpha+2^{n-1}m$.
	If $v_2(d)\le v_2(2\alpha)$,
	then $v_2\left(\frac{4\alpha}{d}\right)=v_2(2\alpha)+1-v_2(d)\ge 1$, and hence $\frac{2\alpha}{d}$ is an integer, which contradicts with $d\nmid 2\alpha$.
	Thus $v_2(2\alpha)<v_2(d)=v_2(4\alpha)\le v_2(2^n m)=n$.
	If $v_2(\alpha)=n-2$, then $v_2(2\alpha)=n-1$, $v_2(d)=n$, thus $v_2(2\alpha+2^{n-1}m)\ge n$. A similar argument on $v_2\left(\frac{4\alpha+2^nm}{d}\right)$ gives $d \mid (2\alpha+2^{n-1}m)$, a contradiction.
	Hence $v_2(\alpha)\le n-3$, or equivalently $v_2(2\alpha)\le n-2$.

	Since $v_2(\alpha)<v_2(2\alpha)<v_2(d)\le v_2(\alpha+\beta)$, then $v_2(\beta)=v_2(\alpha+\beta-\alpha)=\min\{v_2(\alpha),v_2(\alpha+\beta)\}=v_2(\alpha)$, and
	$v_2(\alpha-\beta)=v_2(2\alpha-(\alpha+\beta))=\min\{v_2(2\alpha),v_2(\alpha+\beta)\}=v_2(2\alpha)$.

	The other direction can be proved similarly.
\end{proof}

We can get more results.
\begin{thm}
	Consider the case $M=[-3,5]^*$. Nonsingular perfect $B[-3,5](q)$ splitter sets exist if and only if $\pm 4\notin\langle 6,16\rangle$ and $\ord_q(-\frac45)$ is odd.
\end{thm}
\begin{proof}Let $q=2^nm+1$ such that there exists a perfect $B[-3,5](q)$ set, where $2\nmid m$, $n\ge 3$, then
	$\ind_g(-1)=2^{n-1}m\equiv 2^{n-1}\pmod{2^n}$. Then
	$k=n,a_{001}=\ind_g(-1)$, $a_{000}=0$, and hence the $2$-adic valuations of the remaining 6 elements
	of $A$ must be smaller than $n-1$, especially, $$v_2(\ind_g(2))=v_2(\ind_g(-2))<n-1,$$ $$v_2(\ind_g(3))=v_2(\ind_g(-3))<n-1.$$
	Since $\ind_g(4)\equiv 2\ind_g(2)\pmod{(q-1)}$ and $v_2(\ind_g(2))<v_2(q-1)-1$, $$v_2(\ind_g(4))=v_2(\ind_g(2))+1.$$
	So we have $$\begin{array}{lrcl}
			 & \ind_g(4)                & \in & \{a_{010},a_{011}\},            \\
			 & \{\ind_g(2),\ind_g(-2)\} & =   & \{a_{100},a_{101}\} \text{ or }
			\{a_{110},a_{111}\},                                                \\\text{and } &\{\ind_g(3),\ind_g(-3)\}&=&\{a_{100},a_{101},a_{110},a_{111}\}\\ &&&\setminus \{\ind_g(2),\ind_g(-2)\}.\end{array}$$ Hence we must have $\{\ind_g(4),\ind_g(5)\}=\{a_{010},a_{011}\}$. That is, we require that $$\begin{array}{lrcl}
			                  & v_2(\ind_g(2))      & =          & v_2(\ind_g(3)),           \\
			                  & \ind_g(3)-\ind_g(2) & \equiv     & \ind_g(4)-0\pmod{2^{i_2}} \\
			\text{and }       &
			\ind_g(4)+2^{n-1} & \equiv              & \ind_g(-4)
			\\&&\equiv&\ind_g(5)\pmod{2^n}.\end{array}$$
	As before, the first two equations hold if and only if $\pm 4\notin\langle 6,16\rangle$.We point out that
	$\ind_g(-4)\equiv \ind_g(5)\pmod{2^n}$ is equivalent to the assertion that $\ord_q(-\frac45)$ is odd. So the necessary condition is:
	$\pm 4\notin\langle 6,16\rangle$ and $\ord_q(-\frac45)$ is odd. It is easy to verify that it is also the sufficient condition.
\end{proof}
\begin{exampl}
	Let $q=97$. It was shown in \textnormal{\cite[Example III.1]{ye2020some}} that there exists a perfect $B[-4,4](97)$ set, that is, $\pm 4\notin\langle 6,16\rangle$.
	However, we can verify that $-\frac45$ is a primitive root modulo $97$ and hence its multiplicative order cannot be odd. So there does not exist a perfect $B[-3,5](97)$ set.

	Now let $q=12721=2^4 \cdot 3 \cdot 5\cdot 53+1$, $g=13$, then $$
		\begin{array}{lll}
			\ind_g(2)=1570, &
			\ind_g(3)=1934, & \ind_g(4)=3140, \\\ind_g(-4)=9500,&\ind_g(6)=3504,&\ind_g(16)=6280,\\
			\ord_q(-\frac45)=265.
		\end{array}$$
	Since $\gcd(3504,6280,12720)=8, 8\nmid 3140,8\nmid 9500$, we see that there exists a perfect $B[-3,5](12721)$ set
	$$
		\{g^{16i+j}:i\in[0,794],j\in[0,1]\}.
	$$
	Other primes that satisfy the condition include $26641$, $34729$, $49369$, and $78241$.
\end{exampl}
\begin{thm}
	We now consider the case $M=[-2,6]^*$.
	Nonsingular perfect $B[-2,6](q)$ splitter sets exist if and only if both $\ord_q(-\frac56)$ and $\ord_q(-\frac34)$ are odd and $v_2(\ind_g(3))=v_2(\ind_g(2))+1<v_2(q-1)-1$.
\end{thm}
\begin{proof}
	Let $q=2^nm+1$, $2\nmid m$, $n\ge 3$, then
	$\ind_g(-1)=2^{n-1}m\equiv 2^{n-1}\pmod{2^n}$. Hence
	$k=n, a_{001}=\ind_g(-1)$, $a_{000}=0$ as before. The $2$-adic valuations of the remaining 6 elements
	of $A$ must be smaller than $n-1$. Especially, $v_2(\ind_g(2))<n-1$.
	Since $\ind_g(4)\equiv 2\ind_g(2)\pmod{(q-1)}$, $$v_2(\ind_g(4))=v_2(\ind_g(2))+1.$$
	So we must have
	$$\begin{array}{lrcl}
			 & \ind_g(4)                & \in & \{a_{010},a_{011}\},            \\
			 & \{\ind_g(2),\ind_g(-2)\} & =   & \{a_{100},a_{101}\} \text{ or }
			\{a_{110},a_{111}\}.\end{array}$$
	That is, $i=v_2(\ind_g(2))+1,
		j=v_2(\ind_g(4))+1=i+1\le n-1$.
	Since $\ind_g(6)\equiv \ind_g(3)+\ind_g(2)\pmod{(q-1)}$,
	$$\ind_g(6)\equiv \ind_g(3)+2^{i-1}\pmod{2^{i}}.$$

	If $v_2(\ind_g(3))=i-1$, then $v_2(\ind_g(6))>i-1$ and hence
	$v_2(\ind_g(6))=j-1$, $\{\ind_g(4),\ind_g(6)\}=\{a_{010},a_{011}\}$. Then $\ind_g(6)\equiv \ind_g(4)+2^{k-1}\pmod{2^k}$.
	On the other hand $\ind_g(6)-\ind_g(4)\equiv \ind_g(3)-\ind_g(2)\pmod {(q-1)}$. Hence $\ind_g(3)\equiv \ind_g(2)+2^{k-1}\equiv \ind_g(-2)\pmod{2^k}$, a contradiction.

	If $v_2(\ind_g(3))=j-1$, then $\{\ind_g(4),\ind_g(3)\}=\{a_{010},a_{011}\}$. So $\ord_q(-\frac34)$ is odd. Hence $$\{\ind_g(5),\ind_g(6)\}
		=\{a_{100},a_{110},a_{101},a_{111}\}\setminus\{\ind_g(2),\ind_g(-2)\}.$$ Notice that
	since $\ind_g(6)\equiv 2^{j-1}+\ind_g(2)\pmod{2^j}$, we can also ensure that $\ind_g(6)\notin\{\ind_g(2),\ind_g(-2)\}$.
	So we must have $\ind_g(5)\equiv \ind_g(6)+2^{k-1}\pmod{2^k}$. That is, we need $\ord_q(-\frac56)$ and $\ord_q(-\frac34)$ to be odd and $v_2(\ind_g(3))=v_2(\ind_g(2))+1<n-1$.
	It is easy to verify that it is also the sufficient condition.
	\label{-2-6}
\end{proof}
\begin{exampl}
	Let $q=307009$, then $g=7$,
	$\ind_g(2)=280522=2\cdot 11\cdot 41\cdot 311$,
	$\ind_g(3)=13424=2^2\cdot 33581$,
	$\ord_q(-\frac34)=1599$,
	$\ord_q(-\frac56)=369$.
	Then
	$$
		\{g^{64i+8j+k}:i\in[0,4796],j\in[0,3],k\in[0,1]\}
	$$
	is a perfect $B[-2,6](307009)$ set.
	Other primes that satisfy the condition include $315361$, $348769$, $438769$, and  $442609$.
\end{exampl}
\begin{remark}
	Since $q\equiv1\pmod 8$,
	$2$ is a quadratic residue modulo $q$, hence $v_2(\ind_g(2))\ge 1$.
	So if there is a perfect $B[-2,6](q)$ set, then $v_2(q-1)\ge 4$, $q\equiv 1\pmod{16}$.
\end{remark}
\begin{thm}
	This time we consider the case $M=[-1,7]^*$.
	Nonsingular perfect $B[-1,7](q)$ splitter sets exist if and only if one of the following three conditions holds:
	\begin{enumerate}\item both $\ord_q(-\frac23)$ and $\ord_q(-\frac57)$ are odd, $v_2(\ind_g(2))=v_2(\ind_g(3))$,
		      $v_2(\ind_g(5)-\ind_g(2))=v_2(\ind_g(4))<v_2(q-1)-1$;

		\item $v_2(\ind_g(4))<v_2(q-1)-1$, and all of $\ord_q(-\frac34)$, $\ord_q(-\frac67)$ and $\ord_q(-\frac25)$ are odd;

		\item $v_2(\ind_g(4))<v_2(q-1)-1$, and all of $\ord_q(-\frac34)$, $\ord_q(-\frac27)$ and $\ord_q(-\frac56)$ are odd.
	\end{enumerate}
\end{thm}
\begin{proof}
	Let $q=2^nm+1$, $2\nmid m$, $n\ge 3$, then
	$\ind_g(-1)=2^{n-1}m\equiv 2^{n-1}\pmod{2^n}$. As before,
	$k=n, a_{001}=\ind_g(-1)$, $a_{000}=0$, and the $2$-adic valuations of the remaining 6 elements
	of $A$ must be smaller than $n-1$. As before,
	$$\begin{aligned}\ind_g(4) & \in\{a_{010},a_{011}\},                 \\
               \ind_g(2) & \in\{a_{100},a_{101},a_{110},a_{111}\}.\end{aligned}$$ That is, $i=v_2(\ind_g(2))+1,
		j=v_2(\ind_g(4))=i+2<n$.
	Since $\ind_g(6)\equiv \ind_g(3)+\ind_g(2)\pmod{(q-1)}$,
	$$\ind_g(6)\equiv \ind_g(3)+2^{i-1}\pmod{2^{i}}$$
	always holds.

	If $v_2(\ind_g(3))=i-1$, then $v_2(\ind_g(6))>i-1$ and hence
	$v_2(\ind_g(6))=j-1$. Then $\ind_g(6)\equiv \ind_g(4)+2^{k-1}\pmod{2^k}$.
	On the other hand, $$\ind_g(6)-\ind_g(4)\equiv \ind_g(3)-\ind_g(2)\pmod {(q-1)}.$$ Hence $$\ind_g(3)\equiv \ind_g(2)+2^{k-1}\pmod{2^k}.$$ Then we must have
	$\ind_g(5)\equiv \ind_g(7)+2^{k-1}\pmod{2^k}$ and
	$\ind_g(5)\equiv\ind_g(2)+2^{j-1}\pmod{2^{j}}$. It is equivalent to the assertion that:
	both $\ord_q(-\frac23)$ and $\ord_q(-\frac57)$ are odd, $v_2(\ind_g(2))=v_2(\ind_g(3))$,
	and $v_2(\ind_g(5)-\ind_g(2))=v_2(\ind_g(4))<v_2(q-1)-1$.

	If $v_2(\ind_g(3))=j-1$, then $\ind_g(6)\equiv 2^{j-1}+\ind_g(2)\pmod{2^{j}}$. So $v_2(\ind_g(6))=i-1$.
	And we must have $v_2(\ind_g(5))=v_2(\ind_g(7))=i-1$.

	Suppose $\ind_g(2)=a_{100}$, then $\ind_g(4)=a_{010}, \ind_g(3)=a_{011}$, $\ind_g(6)\in\{a_{110},a_{111}\}$.

	\begin{itemize}
		\item
		      If $\ind_g(6)=a_{110}$, then $\{\ind_g(5),\ind_g(7)\}=\{a_{101},a_{111}\}$.
		      We have $\ind_g(3)\equiv \ind_g(4)\pmod{2^k}$ (which is equivalent to saying that $\ord_q(-\frac34)$ is odd),
		      $\ind_g(6)\equiv \ind_g(2)+2^{j-1}\pmod{2^j}$ (which automatically holds),
		      $\lfloor \ind_g(6)/2^{n-1}\rfloor\equiv 0\pmod{2}$.
		      \begin{itemize}
			      \item
			            If $\ind_g(5)=a_{101}$, then
			            $\ind_g(5)\equiv \ind_g(2)+2^{k-1}\pmod{2^k}$ (which is equivalent to saying that $\ord_q(-\frac25)$ is odd),
			            $\ind_g(7)\equiv \ind_g(6)+2^{k-1}\pmod{2^k}$ (which is equivalent to saying that $\ord_q(-\frac67)$ is odd).
			      \item
			            If $\ind_g(7)=a_{101}$, then
			            $\ind_g(7)\equiv \ind_g(2)+2^{k-1}\pmod{2^k}$ (which is equivalent to saying that $\ord_q(-\frac27)$ is odd),
			            $\ind_g(5)\equiv \ind_g(6)+2^{k-1}\pmod{2^k}$ (which is equivalent to saying that $\ord_q(-\frac56)$ is odd).
		      \end{itemize}
		\item
		      If $\ind_g(6)=a_{111}$, then $\{\ind_g(5),\ind_g(7)\}=\{a_{110},a_{101}\}$.
		      We have $\ind_g(3)\equiv \ind_g(4)\pmod{2^k}$ (which is equivalent to saying that $\ord_q(-\frac34)$ is odd),
		      $\ind_g(6)\equiv \ind_g(2)+2^{j-1}\pmod{2^j}$ (which automatically holds),
		      $\lfloor \ind_g(6)/2^{n-1}\rfloor\equiv 1\pmod{2}$.

		      \begin{itemize}
			      \item
			            If $\ind_g(5)=a_{101}$, then
			            $\ind_g(5)\equiv \ind_g(2)+2^{k-1}\pmod{2^k}$ (which is equivalent to saying that $\ord_q(-\frac25)$ is odd),
			            $\ind_g(7)\equiv \ind_g(6)+2^{k-1}\pmod{2^k}$ (which is equivalent to saying that $\ord_q(-\frac67)$ is odd).

			      \item
			            If $\ind_g(7)=a_{101}$, then
			            $\ind_g(7)\equiv \ind_g(2)+2^{k-1}\pmod{2^k}$ (which is equivalent to saying that $\ord_q(-\frac27)$ is odd),
			            $\ind_g(5)\equiv \ind_g(6)+2^{k-1}\pmod{2^k}$ (which is equivalent to saying that $\ord_q(-\frac56)$ is odd).
		      \end{itemize}
	\end{itemize}
	We find that the condition on $\lfloor \ind_g(6)/2^{n-1}\rfloor$ is redundant. The cases that $\ind_g(2)$ is one of the other elements of $\{a_{100},a_{101},a_{110},a_{111}\}$ can be coped with similarly, and produce the same condition.
	So we get the desired necessary condition. It is easy to verify that it is also the sufficient condition.
\end{proof}
\begin{exampl}
	We give examples for each of the three cases:
	\begin{enumerate}
		\item Let $q=475729=2^4\cdot 3\cdot 11\cdot 17\cdot 53+1$, then $g=13$,
		      $\ind_g(2)=3786$,
		      $\ind_g(3)=54066$,
		      $\ind_g(4)=7572$,
		      $\ind_g(5)=185182$,
		      $\ord_q(-\frac23)=9911$,
		      $\ord_q(-\frac57)=9911$, and
		      $$
			      \begin{aligned}
				        & \{\ind_g(i)\pmod{16}:i\in[-1,7]^*\} \\
				      = & \{8,0,10,2,4,14,12,6\}.
			      \end{aligned}
		      $$
		      Then
		      $$
			      \{g^{16i+j}:i\in[0,29732],j\in[0,1]\}
		      $$
		      is a perfect $B[-1,7](475729)$ set.
		\item Let $q=2693329=2^4\cdot 3\cdot 11\cdot 5101+1$, $g=13$, then
		      $\ord_q(-\frac25)=56111$,
		      $\ord_q(-\frac34)=56111$,
		      $\ord_q(-\frac67)=168333$,
		      and
		      $$
			      \begin{aligned}
				        & \{\ind_g(i)\pmod{16}:i\in[-1,7]^*\} \\
				      = & \{8,0,14,4,12,6,2,10\}.
			      \end{aligned}
		      $$
		      Then
		      $$
			      \{g^{16i+j}:i\in[0,168332],j\in[0,1]\}
		      $$
		      is a perfect $B[-1,7](2693329)$ set.

		\item Let $q=861361=2^4\cdot 3\cdot 5\cdot 37\cdot 97+1$, $g=11$, then
		      $\ord_q(-\frac56)=10767$,
		      $\ord_q(-\frac34)=53835$,
		      $\ord_q(-\frac27)=10767$,
		      and
		      $$
			      \begin{aligned}
				        & \{\ind_g(i)\pmod{16}:i\in[-1,7]^*\} \\
				      = & \{8,0,2,12,4,6,14,10\}.
			      \end{aligned}
		      $$
		      Then
		      $$
			      \{g^{16i+j}:i\in[0,53834],j\in[0,1]\}
		      $$
		      is a perfect $B[-1,7](861361)$ set.
	\end{enumerate}
\end{exampl}
\begin{remark}
	As in the case $[-2,6]^*$,
	if there is a perfect $B[-1,7](q)$ set then
	$q\equiv 1\pmod{16}$.
\end{remark}
\begin{remark}
	It is proved in \textnormal{\cite{yuan2022nonsingular}} that there are infinitely many primes $p$ such that perfect $B[-k_1,k_2](p)$ sets exist for $(k_1,k_2)=(-1,7),(-2,6),(-3,5),(-4,4)$. The argument there can be concluded as: to provide a $B[-k_1,k_2](p)$ set $B$, we construct  by manual trials a homomorphism $f:\langle -1,2,\ldots,k_2\rangle\to (\mathbb{Z}_k,+)$, which is injective when restricted to $[-k_1,k_2]^*$ and satisfies certain other conditions. A number theoretical result tells us that under such conditions there are infinitely many primes $p$ such that $f$ induces a homomorphism $\mathbb{Z}_p^\times\to\mathbb{Z}_k$ and then $\mathbb{Z}_p^\times$ will be split by $[-k_1,k_2]^*$. Although not obvious, primes guaranteed by this argument do satisfy our existence conditions. The verification is omitted for brevity.
\end{remark}
\section{A connection between $B[-k+1,k+1](q)$ sets and $B[-k,k](q)$ sets}
In this section we point out a relation between certain splitter sets. Such connection may be beneficial for looking for existence conditions of splitter sets. Especially, in the next section we will show how such an idea helps us derive a necessary and sufficient condition for the existence of $B[-1,5](q)$ sets.

\begin{prop}
	$B$ is a perfect $B[-k+1,k+1](q)$ splitter set if and only if $B$ is a perfect $B[-k,k](q)$ splitter set and $-\frac{k}{k+1}\in\pi(B)$.
	\label{k-(k+1)}
\end{prop}
\begin{lemma}[{\cite[Lemma III.1]{ye2020some}}]
	If $B$ is a perfect $B[-k+1,k+1](q)$ splitter set, then
	$-\frac{k}{k+1}\in\pi(B)$.
\end{lemma}
\begin{proof}[Proof of Proposition \ref{k-(k+1)}]
	If $B$ is a perfect $B[-k+1,k+1](q)$ splitter set, then $(k+1)B=(k+1)(-\frac{k}{k+1}B)=-kB$, since $-\frac{k}{k+1}\in\pi(B)$. Hence $\{-kB,(-k+1)B,\ldots,-B,B,\ldots,kB\}$ is a partition of $\mathbb{Z}_q^\times$, which means
	$B$ is a perfect $B[-k,k](q)$ splitter set. The converse direction is proved analogously.
\end{proof}
\begin{exampl}
	Let us see how to get the existence condition
	of perfect $B[-1,3](q)$ splitter sets from
	that of perfect $B[-2,2](q)$ sets. From Proposition \ref{k-(k+1)},
	a perfect $B[-1,3](q)$ set exists if and only if there exists a perfect $B[-2,2](q)$ set $B$ such that $\pi(B)$ contains $-\frac{2}3$. We know that a perfect $B[-2,2](q)$ set exists if and only if $v_2(\ord_q(2))\ge 2$. If it is the case,
	then by Corollary \ref{-2-2} and Theorem \ref{thm-prd}, $v_2([G:\pi(B)])\ge v_2(q-1)$, hence the multiplicative order of every element in $\pi(B)$ must be odd. On the other hand, if $\ord_q(-\frac{2}3)$ is odd, then we can construct as in the proof of Theorem \ref{thm-dir-fct} a perfect $B[-2,2](q)$ set $B$ whose stable group contains $-\frac23$, and then by the proof of  Proposition \ref{k-(k+1)}, $B$ is also a perfect $B[-1,3](q)$ set.
\end{exampl}
\begin{exampl}[{\cite[Theorem III.1]{ye2020some}}]
	Let us see $B[-2,4](q)$ sets. The necessary and sufficient condition for the existence of perfect $B[-3,3](q)$ sets is $2\notin \langle 6,8\rangle$. To ensure that $-\frac34$ is in $\pi(B)$, we need $-1\notin\langle-\frac34\rangle$, that is,
	$\ord_q(-\frac{3}4)$ must be odd. So the necessary condition for the existence of perfect $B[-2,4](q)$ sets is that $2\notin \langle 6,8\rangle$ and $\ord_q(-\frac{3}4)$ is odd. It is easy to verify that it is also the sufficient condition.
\end{exampl}

\begin{exampl}
	The results on $B[-4,4](q),B[-3,5](q)$ sets in the previous section also follow this pattern. We omit the proof here.
\end{exampl}
\section{On perfect $B[-1,5](q)$ splitter sets}
In \cite{yuan2022nonsingular} a necessary condition for the existence of perfect $B[-1,5](q)$ splitter sets was given. In this section, we give a sufficient and necessary condition.
\begin{lemma}[\cite{yuan2022nonsingular}]
	If $B$ is a perfect $B[-1,5](q)$ splitter set, then either $\langle-\frac23,-\frac45\rangle\subset \pi(B)$ or $\langle-\frac25,-\frac43\rangle\subset \pi(B)$.
\end{lemma}
Similar to the proof of Proposition \ref{k-(k+1)} we can prove the following proposition:
\begin{prop}
	If $B$ is a perfect nonsingular $B[-1,5](q)$ splitter set, then
	$\mathbb{Z}_q^\times=\{\pm1,\pm2,\pm4\}\cdot B$ is a splitting.
	On the contrary, if there exists a splitting set $B$ of $\{\pm1,\pm2,\pm4\}$,
	and either $\langle-\frac23,-\frac45\rangle\subset \pi(B)$ or $\langle-\frac25,-\frac43\rangle\subset \pi(B)$, then there exists a perfect nonsingular $B[-1,5](q)$ splitter set.
\end{prop}
\begin{thm}
	Let $q=3^k2^lm+1$ be a prime, $k\ge 1$, $l\ge 1$, $\gcd(6,m)=1$, then a perfect $B[-1,5](q)$ set exists if and only if one of the following conditions holds:
	\begin{enumerate}
		\item $v_3(\ind_g(2))<k$, $v_3(\ind_g(2))<v_3(\ind_g(-\frac23))$, $v_3(\ind_g(2))<v_3(\ind_g(-\frac45))$, and both $\ord_q(-\frac23)$ and $\ord_q(-\frac45)$ are odd;
		\item $v_3(\ind_g(2))<k$, $v_3(\ind_g(2))<v_3(\ind_g(-\frac25))$, $v_3(\ind_g(2))<v_3(\ind_g(-\frac34))$, and both $\ord_q(-\frac25)$ and $\ord_q(-\frac34)$ are odd.
	\end{enumerate}
\end{thm}
\begin{proof}
	Note that $\ind_g(4)\equiv 2\ind_g(2)\pmod{(q-1)}$, and the $3$-adic valuations of $\ind_g(2)$ and $\ind_g(4)$ are always the same. Hence a splitting set $B$ of $\{\pm1,\pm2,\pm4\}$ exists as long as $v_3(\ind_g(2))<k$, and we need only the condition on $\pi(B)$.

	If such $B$ exists, we must have $-1\notin \pi(B)$. Without loss of generality, suppose that $\langle-\frac23,-\frac45\rangle\subset \pi(B)$. Hence $-1\notin \langle -\frac23\rangle$ and $-1\notin\langle -\frac45\rangle$, so $\ord_q(-\frac23),\ord_q(-\frac45)$ must be odd. Moreover, since $\langle-\frac23,-\frac45\rangle\subset \pi(B)$ and $3^{v_3(\ind_g(2))+1}\mid [\mathbb{Z}_q^\times:\pi(B)]$ by Theorem \ref{thm-prd}, we have $ v_3(\ind_g(2))<v_3(\ind_g(-\frac23))$, $v_3(\ind_g(2))<v_3(\ind_g(-\frac45))$.

	For the other direction, we have to construct $B$ which contains $\langle-\frac23,-\frac45\rangle$ or $\langle-\frac25,-\frac43\rangle$. For example, suppose that $i_n=v_3(\ind_g(2))+1\le k, v_3(\ind_g(2))<v_3(\ind_g(-\frac23))$, $v_3(\ind_g(2))<v_3(\ind_g(-\frac45))$, and both $\ord_q(-\frac23)$ and $\ord_q(-\frac45)$ are odd.

	Set $$\begin{aligned}C= & \{0,3^{k}2^l,\ldots,3^{k}2^l(m-1)\} \\&+
               \{0,3^{i_n}2^l,2\cdot 3^{i_n}2^l,\ldots,
               (3^k-3^{i_n})2^l\}                       \\&+\{0,3^{i_n},\ldots,3^{i_n}(2^{l-1}-1)\}\\&+\{0,1,\ldots,3^{i_n-1}-1\},\end{aligned}$$ then we can verify, as in the proof of Theorem \ref{thm-dir-fct}, that $$\{0,\ind_g(2),\ind_g(4)\}+\{0,3^{i_n}2^{l-1}\}+C=\mathbb{Z}_{q-1},$$ and $\ind_g(-\frac23),\ind_g(-\frac45)\in C$. Thus $B=\{g^i:i\in C\}$ is the splitter set we need.
\end{proof}
\begin{exampl}
	We give examples for each of the two cases:
	\begin{enumerate}
		\item Let $q=463=2\cdot 3\cdot 7\cdot 11+1$.
		      Then $g=3$, $\ind_g(2)=34=2\cdot 17$, $\ind_g(3)=1$,
		      $\ind_g(-\frac23)=264=2^3\cdot3\cdot11$,
		      $\ind_g(-\frac45)=174=2\cdot3\cdot29$,
		      $\ord_q(-\frac23)=7$,
		      $\ord_q(-\frac45)=77$, and
		      $$
			      \begin{aligned}
				        & \{\ind_g(i)\pmod{6}:i\in[-1,5]^*\} \\
				      = & \{3,0,4,1,2,5\}.
			      \end{aligned}
		      $$
		      Then
		      $$
			      \{g^{6i}:i\in[0,76]\}
		      $$
		      is a perfect $B[-1,5](463)$ set.

		      Notice that $|B|=77$, $\gcd(6,|B|)=1$, so \textnormal{\cite[Theorem 10]{zhang2018nonexistence}} is also applicable to the example above. Next prime that satisfies 1) is $q=1489=2^4\cdot 3\cdot 31+1$, and this time \textnormal{\cite[Theorem 10]{zhang2018nonexistence}} is not applicable. This time
		      $g=14$, $\ind_g(2)=1222=2\cdot13\cdot47$,
		      $\ind_g(3)=190=2\cdot5\cdot19$, $\ind_g(-\frac23)=288=2^5\cdot3^2$,
		      $\ind_g(-\frac45)=960=2^6\cdot3\cdot5$, and
		      $$
			      \begin{aligned}
				        & \{\ind_g(i)\pmod{48}:i\in[-1,5]^*\} \\
				      = & \{24,0,22,46,44,20\}.
			      \end{aligned}
		      $$
		      Then it can be verified that
		      $$
			      \{g^{48i+3j}:i\in[0,30],j\in[0,7]\}
		      $$
		      is a perfect $B[-1,5](1489)$ set.
		      Other primes that satisfy the condition 1) include $2503$, $3583$, $5407$ and $5647$.
		\item The minimal prime that satisfies the condition 2) is $7$, which is a trivial case. Next such prime is $571=2\cdot 3\cdot 5\cdot 19+1$, and \textnormal{\cite[Theorem 10]{zhang2018nonexistence}} is applicable.
		      Next such prime is $1171=2\cdot 3^2\cdot5\cdot13+1$.
		      This time $g=2$, $\ind_g(2)=1$,
		      $\ind_g(3)=155$,
		      $\ind_g(-\frac25)=1158=2\cdot3\cdot193$,
		      $\ind_g(-\frac34)=2\cdot3^2\cdot41$,
		      $\ord_q(-\frac25)=195$,
		      $\ord_q(-\frac34)=65$, and
		      $$
			      \begin{aligned}
				        & \{\ind_g(i)\pmod{6}:i\in[-1,5]^*\} \\
				      = & \{3,0,1,5,2,4\}.
			      \end{aligned}
		      $$
		      Then
		      $$
			      \{g^{6i}:i\in[0,194]\}
		      $$
		      is a perfect $B[-1,5](1171)$ set.
		      Other primes that satisfy the condition 2) include $2371$, $2539$, $3571$, $11251$ and $11437$.
	\end{enumerate}
\end{exampl}
\section{Nonexistence results on quasi-perfect singular splitter sets}
We provide here a nonexistence result of quasi-perfect $B[0,k](m)$ splitter sets.
\begin{prop}
	If $m>k$ and $k$ divides $m$, then there is no quasi-perfect $B[0,k](km)$ set.
\end{prop}
\begin{proof}
	We arrange the elements of $\mathbb{Z}_{km}\setminus\{0,m,2m,\cdots,(k-1)m\}$ as follows:
	\[
		\begin{array}{rrcrr}
			1,     & 1+m,   & \cdots, & 1+(k-2)m, & 1+(k-1)m \\
			2,     & 2+m,   & \cdots, & 2+(k-2)m, & 2+(k-1)m \\
			\vdots & \vdots & \cdots, & \vdots    & \vdots   \\
			m-1,   & 2m-1,  & \cdots, & (k-1)m-1, & km-1
		\end{array}
	\]
	If multiplied by $k$, then elements of every row will become the same. So any splitter set $B$ intersects every row at at most one element, for otherwise there will exist $i,j\in B$ such that $ki\equiv kj\pmod{km}$, which is a contradiction. Moreover, since $kim\equiv 0\pmod{km}$, $B$ cannot have any element of the form $im$. As a quasi-perfect splitter set has $\lfloor (km-1)/k\rfloor=m-1$ elements, it must have elements from all rows of the array.

	If $m$ is divisible by $k$, then every element of the $k$-th row of the array
	is divisible by $k$. Now if $B$ is a quasi-perfect splitter set, then $kB\cap B=\varnothing$. However, $B$ must have elements from all rows of the array, so $kB=\{k, 2k,\ldots, k(m-1)\}$ is the set of all multiples of $k$ in $\mathbb{Z}_{km}\setminus\{0\}$. So $B$ cannot have any element from the $k$-th row of the array, which is a contradiction.
\end{proof}
The following proposition is a generalization of \cite[Lemma 2]{yari2013some}.
\begin{prop}
	If $k$ has a prime divisor $q$ such that $\gcd(q,m)=1$, then every $B[-(k-1),k](m)$ set is also a $B[-k,k](m)$ set. Therefore,
	if $\lfloor (m-1)/(2k-1)\rfloor > \lfloor (m-1)/2k\rfloor$, then a quasi-perfect $B[-(k-1),k](m)$ set does not exist.
\end{prop}
\begin{proof}
	For $t\in \mathbb{Z}_m\setminus\{0\}$, let $U(t)=t[-(k-1),k]^*$ and $
		U'(t)=t[-k,k]^*$. We only need to prove that if $U(t)\cap U(s)=\varnothing$,
	then $U'(t)\cap U'(s)=\varnothing$. It suffices to prove that $kt\not\equiv -ks\pmod{m}$. Suppose that $kt\equiv -ks\pmod{m}$. Since $q$ is coprime to $m$,
	$(k/q)t\equiv -(k/q)s\pmod{m}$, which contradicts that $U(t)\cap U(s)=\varnothing$.

	Suppose there exists a $B[-(k-1),k](m)$ set $B$, then it is also a $B[-k,k](m)$ set, so $|B|\le \frac{m-1}{2k}$. If $\lfloor (m-1)/(2k-1)\rfloor > \lfloor (m-1)/2k\rfloor$, then $B$ can never be quasi-perfect, hence such quasi-perfect splitter set does not exist.
\end{proof}
\begin{remark}
	The necessary and sufficient condition that $\lfloor (m-1)/(2k-1)\rfloor > \lfloor (m-1)/2k\rfloor$ holds is $m\ge 4k^2-2k+1$, or $m=2k(t+1)-s, t=0,1,\ldots,2k-1, s=0,1,\ldots,t$.
\end{remark}
This theorem tells us when constructing quasi-perfect $B[-(k-1),k](m)$ sets, like the one given in \cite[Theorem IV.4]{ye2020some},
$m$ cannot be too big, unless every prime divisor of $k$ also divides $m$.
\section{Conclusion}
In this paper, we have investigated the existence and nonexistence of perfect and quasi-perfect splitter sets in finite cyclic groups, which have their application in coding theory for flash memory storage. By employing cyclotomic polynomials, we derived a general condition for the existence of group splittings, which extends and unifies preceding  results.

Furthermore, we established a connection between
$B[-k,k](q)$ and
$B[-(k-1),k+1](q)$ splitters, demonstrating how the structure of the stable subgroup
$\pi(B)$ plays a crucial role in determining the existence of such splittings. Our approach also provides a way to the determination of the existence conditions for perfect $B[-1,5](q)$ splitter sets.

Additionally, we presented two nonexistence results for quasi-perfect splitter sets, contributing to a deeper understanding of their structural limitations and providing insights into the constraints on their construction.

\section*{Acknowledgements}
The first author thanks Professor Pingzhi Yuan for helpful discussions and access to their unpublished work.
\bibliographystyle{IEEEtrans}
\bibliography{splitter.bib}
\end{document}